\documentclass[aps,onecolumn,showpacs,superscriptaddress,floatfix]{revtex4-2}

\usepackage{amsmath,amssymb,amsthm}
\usepackage{mathtools}
\usepackage{booktabs}
\usepackage{graphicx}
\usepackage[colorlinks=true,linkcolor=blue!70!black,citecolor=blue!70!black, urlcolor=blue!70!black]{hyperref}
\usepackage{microtype}
\usepackage{enumitem}
\usepackage{xcolor}

\newcommand{\K}{\mathrm{K}}
\newcommand{\RC}{\mathrm{RC}}

\newcommand{\Len}{\mathrm{L}}

\newcommand{\aeq}{\overset{+}{=}}
\newcommand{\ale}{\overset{+}{<}}
\newcommand{\aleq}{\overset{+}{\leq}}
\newcommand{\age}{\overset{+}{>}}
\newcommand{\ageq}{\overset{+}{\geq}}

\newtheorem{theorem}{Theorem}

\newtheorem{definition}[theorem]{Definition}

\begin{document}

\title{\texorpdfstring{The Informational Cost of Structure:\\
       Representational Complexity in Networked Dynamical Systems}{The Informational Cost of Structure:
       Representational Complexity in Networked Dynamical Systems}} 

\author{Cyril Rommens}
\author{Pietro Traversa}%
\affiliation{%
 Institute for Biocomputation and Physics of Complex Systems (BIFI), Universidad de Zaragoza, 50018 Zaragoza, Spain\\
}%
\affiliation{
  Department of Theoretical Physics, University of Zaragoza, 50009 Zaragoza, Spain
}

\author{Guilherme Ferraz de Arruda}%
\affiliation{Instituto de F\'{\i}sica Gleb Wataghin, Universidade de Campinas (UNICAMP), Campinas, Brazil.}

\author{Yamir Moreno}%
\email{yamir@unizar.es}
\affiliation{%
 Institute for Biocomputation and Physics of Complex Systems (BIFI), Universidad de Zaragoza, 50018 Zaragoza, Spain\\
}%
\affiliation{
  Department of Theoretical Physics, University of Zaragoza, 50009 Zaragoza, Spain
}%

\begin{abstract}
How much information is required to represent a dynamical system in terms of an interaction structure and an evolution rule? We address this question using algorithmic information theory. We introduce Representational Complexity, the excess description length of a structure-plus-rule model relative to the shortest possible description of the dynamics itself. This intrinsic description defines a universal lower bound: no exact structural representation can be more concise. If arbitrary rules are allowed, graphs, hypergraphs, and other formalisms can all reach this bound by shifting information between structure and dynamics, so expressiveness alone cannot distinguish them. Meaningful differences arise only when scientific modeling restricts the admissible structures and rules. Within this setting, we identify conditions under which graph and hypergraph descriptions are informationally equivalent, and show how graph-preferred, hypergraph-preferred, and mixed regimes can emerge when those conditions are relaxed. Because Kolmogorov complexity is not computable, we complement the formal results with explicit description-length estimates. Our framework reframes the choice of network representation as a question of informational cost and mechanistic transparency rather than universal expressive power.
\end{abstract}

\maketitle

\section{Introduction}

Structural representations pervade complex-systems science~\cite{newman2018networks, barabasi2013network, Strogatz2001, Costa2011, Boccaletti2006, Vespignani2018}. The choice of structural language --- graph, hypergraph, simplicial complex, etc.~\cite{Torres2021,Battiston2020,BickGross2023,Lambiotte2019, battiston2021physics,Butts2009,Kivela2014,HolmeSaramaki2012}\ --- is often presented as a question of \emph{expressiveness}: can the language represent the interactions of interest?  A recent preprint by Peixoto et al.~\cite{Peixoto2026} argues that graphs are maximally expressive, subsuming hypergraph formulations as constrained special cases, a claim that has been debated on both technical and conceptual grounds~\cite{LaRockLambiotte2025,Lacasa2026,Llabres2026}.

We argue that expressiveness is not the right operational criterion. Any finite Boolean dynamical map is fully specified by its truth table, a finite binary string, and can therefore be encoded exactly within essentially any structural language by absorbing the missing information into the dynamical rule. Expressiveness alone thus cannot motivate a choice between graphs and hypergraphs. A more informative criterion is the \emph{informational cost} of a representation~\cite{rissanen1978modeling, barron1998minimum, Grunwald2007}: given a specific dynamical map $F$, a structural language, and a class of admissible dynamical rules, how much description length does the resulting representation minimally require?

Informational cost is, however, only part of the story. Scientific modeling is not merely the reproduction of input--output relations at minimal bit cost; it is an attempt to identify the relevant degrees of freedom, symmetries, and interaction mechanisms of a system. A fully flexible representation may achieve a short description while concealing the physical hypothesis inside an opaque effective rule. A more constrained representation can be scientifically preferable precisely because it makes that hypothesis explicit and testable. From this viewpoint, hypergraphs are not primarily a claim of greater representational power; they are modeling commitments that declare group-level interactions to be primitive objects~\cite{ferraz2024contagion,stOnge2025defining,abiad2026roadmap}.  The relevant question is therefore not whether a graph with sufficiently flexible node dynamics can emulate a hypergraph model~\cite{Llabres2026}, but whether doing so preserves parsimony, interpretability, identifiability, and mechanistic transparency~\cite{Akaike1974,BellmanAstrom1970, Machamer2000}.

Informational cost and scientific adequacy are complementary lenses: a good model should be both efficiently representable within its chosen language and honest about the mechanisms it posits. Kolmogorov complexity~\cite{LiVitanyi2019} provides a natural language for formalizing the cost side of this distinction. We first identify the theoretical floor: the fully equation-based representation, in which no structural skeleton is provided, and all information is absorbed into the rule, achieves description length $\K(F)$ up to an additive constant. No exact representation can do better (Theorem~\ref{thm:floor}). This result is not meant to advocate equation-based descriptions as scientific models. Rather, it establishes the irreducible algorithmic content of the dynamical map and clarifies that any explicit structure is an additional modeling commitment.

We then distinguish between unrestricted and restricted representations. In an unrestricted representation, the dynamical rule may be any computable program; hence, information absent from the skeleton can always be hidden inside the rule. Such representations capture theoretical expressiveness but are usually poor scientific models because they sacrifice interpretability and generalizability. In a restricted representation, by contrast, both the structure and the dynamics must belong to specified admissible classes. It is in this restricted regime that representations such as graph and hypergraph models can differ meaningfully in informational cost.

Within the restricted setting, we compare graph and hypergraph representations through the relative cost $\Delta\RC=\RC(G,D_G)-\RC(H,D_H)$, whose sign defines graph-preferred, equivalence, and hypergraph-preferred regimes. We prove that the two languages are informationally equivalent up to additive constants \emph{conditionally}: equivalence requires that the higher-order structure be recoverable from the graph through its dynamical rule, and that the graph be recoverable from the hypergraph through a fixed computable projection (Section~\ref{sec:equivalence}). When these conditions fail, the balance of description length can shift. Thus, the equivalence statement is not intended as a generic claim about graph and hypergraph representations, but as a sufficient and interpretable set of conditions under which they carry the same total algorithmic information. Because exact Kolmogorov complexities are incomputable, we explore graph-preferred (Section~\ref{sec:graph_pref}) and hypergraph-preferred (Section~\ref{sec:hyper_pref}) examples through operational description-length estimates based on counting arguments. These estimates build intuition for how informational cost differences can arise, but they do not constitute proofs that the corresponding regimes exist at the level of Kolmogorov complexity; throughout, we keep what is proven separate from what is merely suggested.

Taken together, these results reframe the graph--hypergraph debate. The question is not which language is universally more expressive, but how a given modeling commitment distributes information between structure and dynamics, at what informational cost, and whether that cost is repaid in interpretability, mechanistic transparency, or empirical alignment. Our contributions are as follows:
\begin{itemize}
  \item We introduce Representational Complexity [Eq.~\eqref{eq:rc_def}], an information-theoretic measure of the excess description length incurred when a dynamical map is represented within a chosen structural and dynamical language.
  \item We prove that no exact representation can describe a map below its Kolmogorov complexity $\K(F)$, so that the equation-based representation is a representation-independent floor and any explicit structure is a non-compressive modeling commitment up to additive constants (Theorem~\ref{thm:floor}).
  \item We show that in the unrestricted setting all structural languages are informationally interchangeable, so that informational cost alone cannot motivate a choice between them (Section~\ref{sec:encoding}).
  \item We establish that, in the restricted setting, graph and hypergraph representations are informationally equivalent \emph{conditionally} on explicit recoverability assumptions relating their structural and dynamical components (Section~\ref{sec:equivalence}).
  \item We introduce a computable description-length surrogate for $\RC$ and use it to compare representations across explicit examples, illustrating how graph-preferred and hypergraph-preferred cases may arise in practice (Section~\ref{sec:comparison} onward).
  \item We argue that the choice of structural language should be governed jointly by informational cost and scientific adequacy.
\end{itemize}

\section{Preliminaries and Definitions}
\label{sec:notation}

\subsection{Algorithmic information}

Throughout, all objects are finite binary strings, and we fix a universal prefix-free Turing machine $\mathcal{U}$. We write $\K(x)$ for the prefix-free (Kolmogorov) complexity of $x$, i.e.\ the length of a shortest self-delimiting program that makes $\mathcal{U}$ output $x$~\cite{LiVitanyi2019}. We write $\K(x\mid y^{*})$ for the conditional complexity of $x$ given a shortest program $y^{*}$ for $y$, and $\K(x,y):=\K(\langle x,y\rangle)$ for joint complexity, where $\langle\cdot,\cdot\rangle$ is a fixed self-delimiting pairing of strings.

All $O(1)$ terms are absolute constants depending only on $\mathcal{U}$. By the invariance theorem, replacing $\mathcal{U}$ by another universal machine changes every complexity by at most such a constant; hence every statement below, which holds up to $O(1)$, is machine-independent. We write the additive-constant versions of the relational symbols with a superscript $+$: $f\aeq g$ means $|f-g|=O(1)$, $f\aleq g$ means $f\le g+O(1)$, and likewise for $\ageq,\ale,\age$. The strict symbols $\ale$ and $\age$ assert inequalities whose margin exceeds any fixed additive constant; when used in asymptotic settings, this should be understood as a non-constant separation in description length, typically scaling with system size $N$.

Conditioning on a shortest program $y^{*}$ rather than on $y$ itself is precisely what makes the prefix-free \emph{chain rule}
\begin{equation}
  \K(x,y) \aeq \K(x)+\K(y\mid x^{*})
  \label{eq:chain}
\end{equation}
hold up to an additive constant rather than an additive logarithmic term. Applying Eq.~\eqref{eq:chain} in both orders yields the \emph{symmetry of
information}
\begin{equation}
  \K(x)+\K(y\mid x^{*}) \;\aeq\; \K(y)+\K(x\mid y^{*}) \;\aeq\; \K(x,y).
  \label{eq:soi}
\end{equation}

\subsection{Representational Complexity}
\label{sec:rc}

In complex-systems modeling, one typically describes a dynamical system $F:\{0,1\}^{N}\!\to\{0,1\}^{N}$ by splitting its description into a structural part and a dynamical part: a skeleton $S$ that specifies which units interact, and a rule $D$ that specifies how. Since the domain is finite, we identify $F$ with its truth table, a binary string of length $N2^{N}$, so that $\K(F)$ is the complexity of a finite string.

\begin{definition}[Representation]
\label{def:rep}
A \emph{representation} $R$ of $F$ is a pair $(S,D)$, where $S$ is a finite string encoding an interaction structure and $D$ is a finite string encoding a partial computable function satisfying $D(S)=F$. This includes the requirement that the computation encoded by $D$ halts on the input $S$; throughout, the evaluator that interprets a string as a partial computable function is fixed once and for all. The pair is encoded as a single finite string through the fixed pairing $\langle S,D\rangle$.
\end{definition}

\begin{definition}[Representational Complexity]
\label{def:rc}
The \emph{Representational Complexity} of an exact representation $R=(S,D)$ of $F$ is
\begin{equation}
  \RC(S,D) := \K(S,D)-\K(F).
  \label{eq:rc_def}
\end{equation}
Here $\K(F)=\K(D(S))$ is the complexity of the \emph{dynamical map} produced by the representation, not of the operator $D$: the operator may carry internal information that is invisible in the output map $F$. In particular, $\K(D)$ and $\K(F)$ can differ substantially: a highly structured operator applied to a simple skeleton may produce a complex map, or vice versa.
\end{definition}

Because $F$ is computed from $(S,D)$ by a fixed evaluator that decodes the pair, runs $D$ on $S$, and outputs the result,
\begin{equation} \label{eq:lower_floor}
  \K(F)\aleq\K(S,D),
\end{equation}
so $\RC(S,D)\ageq 0$: no representation can describe $F$ in fewer bits than its intrinsic algorithmic content. This construction instantiates, within the setting of networked dynamical systems, the two-part code framework of algorithmic statistics~\cite{GacsTrompVitanyi2001,VereshchaginVitanyi2004}: the structure $S$ plays the role of the model, and the rule $D$ plays the role of the data-given-model code.

A second decomposition exposes how the cost is split between structure and dynamics. The chain rule~\eqref{eq:chain} gives
\begin{equation}
  \K(S,D)\aeq\K(S)+\K(D\mid S^{*}),
  \label{eq:chain_SD}
\end{equation}
and writing the algorithmic mutual information $I(D:S):=\K(D)-\K(D\mid S^{*})$ turns Eq.~\eqref{eq:rc_def} into
\begin{equation}
  \RC(S,D) \;\aeq\; \underbrace{\bigl[\K(S)+\K(D)-\K(F)\bigr]}_{\text{independent over-cost}} \;-\;I(D:S).
  \label{eq:rc_expanded}
\end{equation}
The bracketed term is what one would pay to describe $S$ and $D$ separately, above the floor $\K(F)$, while $I(D:S)$ is the saving obtained from their algorithmic overlap. Intuitively, $I(D:S)$ measures how much knowing $S$ reduces the description length of $D$: it is large when the skeleton and the rule share algorithmic information (so that one can be inferred, at least partially, from the other), and zero when they are algorithmically independent. Holding the component complexities $\K(S),\K(D),\K(F)$ fixed, $\RC$ therefore decreases monotonically as the dependence between structure and rule increases: it is largest when $S$ and $D$ are algorithmically independent ($I(D:S)\aeq0$), so that the structure resolves none of the uncertainty in the rule and the two parts must be described separately, and smallest when one part nearly determines the other.

\section{Representational Choice and Informational Cost}

We now apply these tools to structural choice. We first show that explicit structure can never reduce the description length of a system below the floor $\K(F)$, so that imposing structure is always a non-negative informational commitment (Section~\ref{sec:K_floor}). We then distinguish unrestricted from restricted representations, showing that a meaningful comparison of informational cost is possible only once the admissible structures and dynamical rules are constrained (Section~\ref{sec:encoding}). Within this restricted setting, we turn to the consequences of committing to a graph or a hypergraph encoding in particular (Section~\ref{sec:comparison}).

\subsection{Imposing Structure Never Decreases Informational Cost}
\label{sec:K_floor}
We first establish the baseline against which structural commitments are measured: the representation that supplies no structure and absorbs everything into the dynamical rule, which attains the lowest possible description length.

\begin{theorem}[Equation-based floor]
\label{thm:floor}
Let $D_{\mathrm{eq}}$ denote a constant function with $D_{\mathrm{eq}}(\varepsilon)=F$, where $\varepsilon$ is the empty string. Then:
\begin{enumerate}[label=(\roman*)]
  \item $\K(\varepsilon,D_{\mathrm{eq}}) \aeq \K(F)$ and $\RC(\varepsilon,D_{\mathrm{eq}}) \aeq 0$.
  \item For any exact representation $(S,D)$ of $F$, $\K(F) \aleq \K(S,D)$, so $\RC(S,D)\ageq 0$.
  \item If $\K(S,D) \aeq \K(F)$, then $\K(S\mid F^*) \aeq 0$: a structure attaining the floor is itself computable from the dynamical map up to additive constants.
\end{enumerate}
\end{theorem}

\begin{proof}
\textit{(i).} Since $\varepsilon$ has constant complexity and $D_{\mathrm{eq}}$ is obtained by wrapping a shortest program for $F$ into a constant function, $\K(\varepsilon,D_{\mathrm{eq}}) \aleq \K(F)$. The reverse inequality follows from Eq.~\eqref{eq:lower_floor}, because evaluating $D_{\mathrm{eq}}$ on $\varepsilon$ gives $F$.

\textit{(ii).} Given a shortest program for the pair $\langle S,D\rangle$, a fixed wrapper program decodes the pair, evaluates $D$ on $S$, and outputs $F$. Therefore there exists a program for $F$ of length $\K(S,D)+O(1)$, proving $\K(F) \aleq \K(S,D)$.

\textit{(iii).} The map $(S,D)\mapsto(F,S)$ is computable: evaluate $D(S)$ to obtain $F$ and output $(F,S)$. Hence $\K(F,S) \aleq \K(S,D) \aeq \K(F)$, and since $\K(F,S)\ageq\K(F)$ trivially, $\K(F,S)\aeq\K(F)$. By the chain rule~\eqref{eq:chain}, $\K(F,S) \aeq \K(F)+\K(S\mid F^*)$, so $\K(S\mid F^*) \aeq 0$.
\end{proof}

Theorem~\ref{thm:floor} should not be read as saying that every non-empty structure is useless. It says something more precise: no structure can reduce the total description length below $\K(F)$, and any zero-RC structure must be algorithmically recoverable from $F$. Consequently, a structure that is not recoverable from the dynamical map necessarily incurs a positive representational cost, up to additive constants. Explicit structure is justified not by beating $\K(F)$, but by providing interpretability, mechanistic organization, or empirical observability.

\subsection{Restricted Structural and Dynamical Representations}
\label{sec:encoding}

Until now, we have treated representations as descriptions of a networked dynamical system that divide their encoding between a structural part $S$ and a dynamical part $D$. In practice, the information cannot be distributed freely between these two components: the choice of a modeling framework imposes specific constraints on what counts as an admissible structure and an admissible rule. We discuss these constraints and their consequences below, beginning with the baseline case in which no constraints are imposed, before turning to restricted model classes.

\subsubsection{The Unrestricted Setting}

\begin{definition}[Unrestricted representation]
\label{def:unrestricted}
A representation $R = (S,D)$ of $F$ is \emph{unrestricted} if $S$ and $D$ are subject to no constraint beyond computability and $D(S)=F$; in particular, no constraint is imposed on how information is distributed between $S$ and~$D$.
\end{definition}

If, in this setting, the structural part $S$ instantiates a specific, familiar structure, like a graph or hypergraph, then this representation is theoretically optimal, attaining the lowest description length that any representation of $F$ can achieve. This optimum is not unique. We have already seen that the fully equation-based representation of Theorem~\ref{thm:floor} attains this floor. Actually, any representation satisfying $\K(S,D) \aeq \K(F)$ attains the same floor. In this sense, Theorem~\ref{thm:floor}(iii) acts as an interchangeability requirement for unrestricted representations. The unrestricted setting therefore serves as a theoretical reference point, establishing $\K(F)$ as the irreducible floor, but representations of this kind generally lack mechanistic interpretability, explanatory structure, and compact causal organization. The meaningful comparison of informational cost begins once the admissible dynamical language is restricted by modeling assumptions.

\subsubsection{The Restricted Setting}

\begin{definition}[Restricted representation]
\label{def:restricted}
A representation $R=(S,D)$ of $F$ is \emph{restricted} if $S$ is constrained to belong to a structural language $\mathcal{L}$ (i.e.\ $S\in\mathcal{L}$) and $D$ is constrained to belong to an admissible class of computable dynamical rules $\mathcal{C}$ (i.e.\ $D\in\mathcal{C}$). The constraints $(\mathcal{L},\mathcal{C})$ fix how information may be distributed between $S$ and $D$. We denote such a representation by $R_\mathcal{M} = (S_\mathcal{L}, D_\mathcal{C})$, where $\mathcal{M}=(\mathcal{L},\mathcal{C})$ is the model class.
\end{definition}

Throughout we assume the model class is \emph{realizable} for the system at hand, i.e.\ that at least one admissible pair represents $F$ exactly: $\{(S,D): S\in\mathcal{L},\,D\in\mathcal{C},\,D(S)=F\}\neq\varnothing$. A restricted representation is \emph{ideal}, denoted $R_{\mathcal{M}}^{\mathrm{ideal}}$, when it attains the smallest Kolmogorov complexity within the model class,
\begin{equation}
  \K(R_{\mathcal{M}}^{\mathrm{\text{ideal}}}) \aeq \min_{\substack{S\in\mathcal{L},\,D\in\mathcal{C}\\D(S)=F}} \K(S,D) \aeq \min_{\substack{S\in\mathcal{L},\,D\in\mathcal{C}\\D(S)=F}} \bigl[\K(S)+\K(D\mid S^*)\bigr].
  \label{eq:restricted_R_ideal}
\end{equation}

The restriction $S\in\mathcal{L}$, $D\in\mathcal{C}$ is a constraint on the admissible representations, not on the reference universal machine: all complexities in Eq.~\eqref{eq:restricted_R_ideal} are ordinary prefix Kolmogorov complexities evaluated after restricting the feasible set to $S\in\mathcal{L}$, $D\in\mathcal{C}$. This should not be confused with MDL-style restricted code lengths, which are introduced later as operational surrogates. The minimum in Eq.~\eqref{eq:restricted_R_ideal} is a set-theoretic minimum over finite strings. It does not require that membership in $\mathcal{L}$ or $\mathcal{C}$ be decidable, nor that the minimizing pair be computable by the modeler. Realizability guarantees that the feasible set is non-empty, and since prefix complexities are non-negative integers, a minimum exists. This non-constructive definition is useful as an ideal benchmark; Section~\ref{sec:approx} introduces explicit code lengths precisely because the ideal minimizer is generally not accessible in practice.

Unlike in the unrestricted setting, the Representational Complexity of an ideal restricted representation need not equal the floor: in general $\RC(R_{\mathcal{M}}^{\mathrm{ideal}})\age 0$, since the constraints $(\mathcal{L},\mathcal{C})$ may forbid the redistribution of information that would be required to reach $\K(F)$. Restricted representations are therefore not informationally interchangeable in general, and comparing their informational cost becomes meaningful. The exception is when two restricted model classes are related by a bijection computable by a single fixed program --- one whose description does not grow with the system. The two classes then have equal Kolmogorov complexity up to an additive constant, placing them in the equivalence regime discussed below. Not every bijection admits such a fixed encoding, however: some require a program whose length scales with the system size, in which case the two classes need not be informationally equivalent. The mere existence of a bijection is thus not sufficient for equivalence; what matters is whether the translation between the classes is itself algorithmically cheap. This restricted setting is where the graph--hypergraph comparison of the next section becomes meaningful.

\subsection{Graph and Hypergraph Representations in the Restricted Setting}
\label{sec:comparison}

We now compare the informational cost of the best achievable graph and hypergraph representations of the same map $F$. Let $\mathcal{M}_G=(\mathcal{L}_G,\mathcal{C}_G)$ and $\mathcal{M}_H=(\mathcal{L}_H,\mathcal{C}_H)$ be the graph and hypergraph model classes, and let $R_G^{\mathrm{ideal}}=(G,D_G)$ and $R_H^{\mathrm{ideal}}=(H,D_H)$ be their respective ideal representations in the sense of Eq.~\eqref{eq:restricted_R_ideal}. 
Here $G$ and $H$ denote particular structural instances selected by the ideal representations, while $\mathcal{L}_G$ and $\mathcal{L}_H$ denote the corresponding structural languages.

The natural object of comparison is the relative representational cost
\begin{align}
\Delta\RC
&= \RC(G,D_G)-\RC(H,D_H)
\nonumber\\
&\aeq \K(G)+\K(D_G\mid G^*) - \K(H)-\K(D_H\mid H^*).
\label{eq:deltaRC_sym}
\end{align}
Because both representations encode the same map, the two floor terms $-\K(F)$ in $\RC(G,D_G)$ and $\RC(H,D_H)$ cancel, and the incomputable quantity $\K(F)$ drops out entirely; the chain rule~\eqref{eq:chain_SD} then splits each joint complexity into a structural and a dynamical part. The comparison thus reduces to a difference of two finite description lengths and depends on two contributions: the complexity of the structural encoding and the complexity of the dynamical rule conditioned on that structure.

The sign of $\Delta\RC$ distinguishes three regimes:
\begin{align}
\Delta\RC &\ale 0 && \text{graph preferred},
\nonumber\\
\Delta\RC &\aeq 0 && \text{informational equivalence},
\nonumber\\
\Delta\RC &\age 0 && \text{hypergraph preferred}.
\nonumber
\end{align}
Note that because all Kolmogorov-complexity statements are defined only up to additive constants, a preference should be interpreted as meaningful only when the separation is larger than this ambiguity; in scaling arguments, this corresponds to a gap that grows with system size. In asymptotic comparisons, we therefore assume that the representation languages and the translators between them are specified uniformly: the decoding conventions, pairing map, evaluator, projection maps, and admissible rule catalogues are fixed finite objects, independent of $N$. Under this standard uniformity convention, all additive constants hidden in $\aeq,\aleq,\ageq$ are bounded independently of system size. Hence, if a regime-defining difference is bounded below by a function $g(N)-O(1)$ with $g(N)\to\infty$, its sign is invariant for all sufficiently large $N$. Conversely, constant-size gaps should not be interpreted as meaningful strict preferences.

\subsubsection{Equivalence regime  \texorpdfstring{($\Delta\RC\aeq0$)}{∆RC = 0}}
\label{sec:equivalence}

By Eq.~\eqref{eq:deltaRC_sym}, graph and hypergraph representations lie in the equivalence regime precisely when
\begin{equation}
\K(G)+\K(D_G\mid G^*) \aeq \K(H)+\K(D_H\mid H^*).
\label{eq:equiv_condition}
\end{equation}
In this regime, neither representation enjoys a compressive advantage: both carry the same algorithmic information, differing only in how that information is distributed between structure and dynamics. Equation~\eqref{eq:equiv_condition} is the exact equivalence condition. We next give a natural pair of sufficient modeling assumptions under which it holds.

Consider a system with irreducible group interactions, such as higher-order contagion~\cite{Battiston2020, FerrazdeArruda2024}. Such a system must specify which subsets of variables are allowed to interact, and the collection of admissible interaction groups naturally defines a hypergraph $H$. A graph-based description, lacking these groups as primitive objects, must instead reconstruct the higher-order structure through its dynamical rule.

\paragraph*{Assumption 1.}(Sufficient condition for equivalence.)
The graph-based rule carries exactly the information needed to recover the hypergraph and then implement the hypergraph dynamics in an information-optimal way:
\begin{equation}
\K(D_G\mid G^*) \aeq \K(H\mid G^*) + \K(D_H\mid H^*).
\label{eq:equiv_assumption1}
\end{equation}
Substituting Eq.~\eqref{eq:equiv_assumption1} into the left-hand side of Eq.~\eqref{eq:equiv_condition} gives
\begin{equation}
\K(G)+\K(D_G\mid G^*) \aeq \K(G)+\K(H\mid G^*)+\K(D_H\mid H^*).
\label{eq:equiv_substituted}
\end{equation}

\paragraph*{Assumption 2.}(Sufficient condition for equivalence.)
The graph can be obtained from the hypergraph through a fixed computable projection rule, so that
\begin{equation}
\K(G\mid H^*) \aeq 0.
\label{eq:equiv_assumption2}
\end{equation}
By symmetry of information~\eqref{eq:soi},
\begin{equation}
\K(G)+\K(H\mid G^*) \aeq \K(G,H) \aeq \K(H)+\K(G\mid H^*),
\label{eq:equiv_symmetry}
\end{equation}
and Assumption~2 removes the final term, leaving $\K(G)+\K(H\mid G^*)\aeq\K(H)$. Substituting this into Eq.~\eqref{eq:equiv_substituted} yields
\begin{equation}
\K(G)+\K(D_G\mid G^*) \aeq \K(H)+\K(D_H\mid H^*),
\nonumber
\end{equation}
which is exactly Eq.~\eqref{eq:equiv_condition}. The two representations are therefore informationally equivalent up to additive constants; Fig.~\ref{fig:G_H_comparison} shows the corresponding decomposition schematically.

\textit{Remark.}
Equation~\eqref{eq:equiv_condition} is the necessary-and-sufficient condition for informational equivalence; Assumptions~1 and~2 are interpretable sufficient conditions that jointly imply it. Thus the derivation should be read primarily as a bookkeeping characterization of one sufficient route to equivalence, not as a claim that arbitrary graph rules decompose in this way. A graph rule may achieve equivalence through mechanisms other than explicitly reconstructing $H$---the assumptions identify one natural route, not the only one. Moreover, both assumptions are non-trivial modeling hypotheses that may fail in most empirical network systems: Assumption~1 requires the graph rule to carry precisely the right amount of higher-order information, while Assumption~2 requires the graph to be recoverable from the hypergraph by a fixed, system-size-independent program. Equivalence is therefore the characterization of a specific regime, not a generic result.

\begin{figure}[h]
    \centering
    \includegraphics[width=0.5\linewidth]{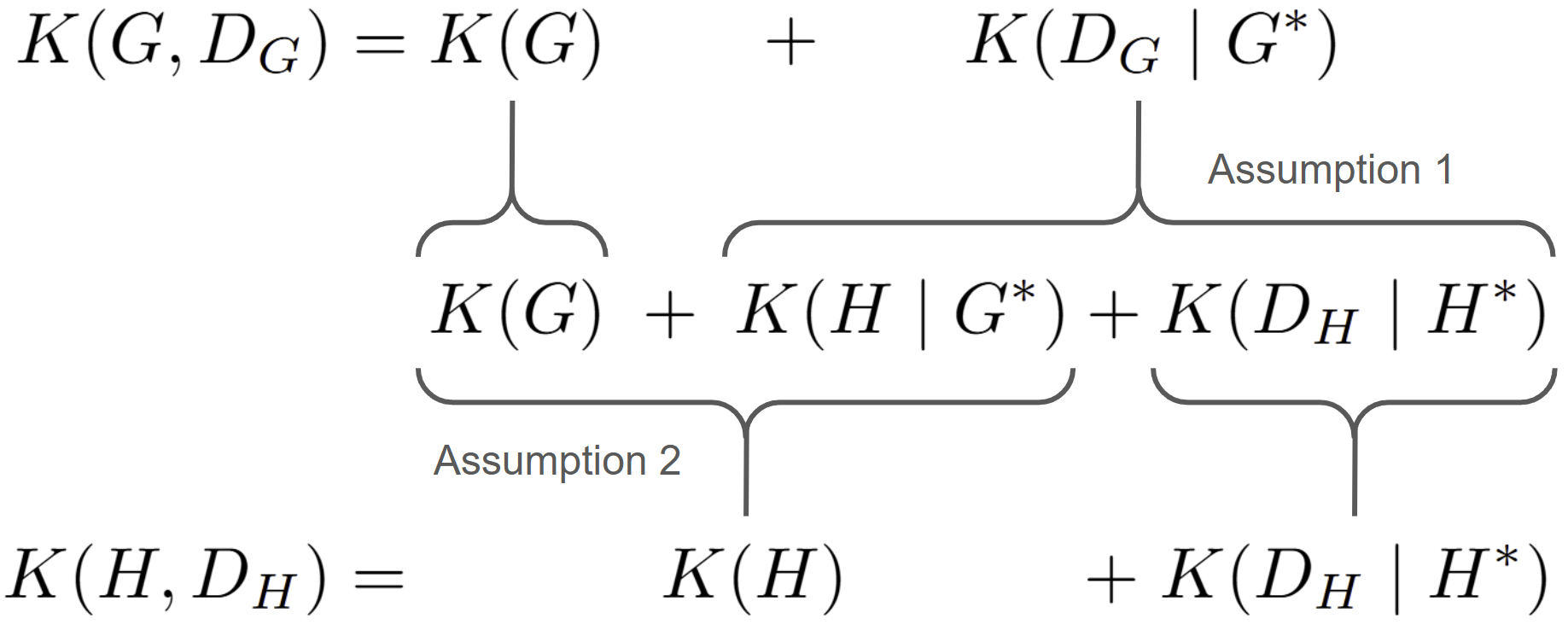}
    \caption{Schematic decomposition of the complexity contributions of the graph and hypergraph representations under Assumptions~1 and~2, which place them in the equivalence regime.}
    \label{fig:G_H_comparison}
\end{figure}

In this regime, informational considerations alone cannot justify a preference between the two descriptions. Any preference must instead rest on non-informational criteria, such as interpretability, mechanistic transparency, empirical observability, or modeling convenience. As an intuitive example, whenever the graph and hypergraph classes are related by a bijection computable by a single fixed program, they encode the same information up to a constant and Eq.~\eqref{eq:equiv_condition} holds. As noted in Section~\ref{sec:encoding}, the existence of a bijection is not by itself sufficient --- it must admit a fixed-length encoding --- but when it does, the two classes lie in the equivalence regime.

The result is conditional on both assumptions. Equation~\eqref{eq:equiv_assumption1} may fail if the higher-order structure compresses differently in the two representations, and Eq.~\eqref{eq:equiv_assumption2} fails when the graph cannot be recovered from the hypergraph by a fixed projection. Relaxing either assumption gives rise to the graph-preferred and hypergraph-preferred regimes considered next.

\subsubsection{Graph-preferred regime \texorpdfstring{($\Delta\RC\ale0$)}{∆RC<0}}
\label{sec:graph_pref}

This regime is obtained from the equivalence case by retaining Assumption~2 but breaking Assumption~1 in the graph's favour. Concretely, suppose the graph still arises as a fixed projection of the hypergraph,
\begin{equation}
\K(G\mid H^*)\aeq0,
\label{eq:graphpref_proj}
\end{equation}
but that the graph-based rule encodes the higher-order information \emph{strictly} more compactly than reconstructing and running the hypergraph,
\begin{equation}
\K(D_G\mid G^*)\ale\K(H\mid G^*)+\K(D_H\mid H^*),
\label{eq:graphpref_strict}
\end{equation}
that is, Eq.~\eqref{eq:equiv_assumption1} fails as an equality and holds only as a strict inequality by more than an additive constant. This is what we mean by breaking Assumption~1 in the graph's favour.

Expanding $\K(G,D_G)$ by the chain rule and applying Eq.~\eqref{eq:graphpref_strict},
\[
\K(G,D_G)\aeq\K(G)+\K(D_G\mid G^*)\ale\K(G)+\K(H\mid G^*)+\K(D_H\mid H^*).
\]
Symmetry of information~\eqref{eq:soi} gives $\K(G)+\K(H\mid G^*)\aeq\K(H)+\K(G\mid H^*)$, and the retained projection assumption~\eqref{eq:graphpref_proj} removes the final term, leaving $\K(G)+\K(H\mid G^*)\aeq\K(H)$. Hence
\[
\K(G,D_G)\ale\K(H)+\K(D_H\mid H^*)\aeq\K(H,D_H),
\]
so $\Delta\RC\ale0$ and the graph representation provides the shorter description.

The interpretation is that the higher-order information required to reproduce the dynamics is captured more compactly inside the graph-based rule than through an explicit hyperedge list together with its own rule. The explicit groups in $H$ then buy no compensating reduction in dynamical cost, and the graph language functions as a shorter generative description of the same dynamical map.

\subsubsection{Hypergraph-preferred regime \texorpdfstring{($\Delta\RC\age0$)}{∆RC>0}}
\label{sec:hyper_pref}

The previous two regimes retained Assumption~2, that the graph is a fixed projection of the hypergraph. The hypergraph-preferred regime is the one in which this assumption \emph{fails}: the operative higher-order groups are not recoverable from the pairwise skeleton by a fixed program, so
\begin{equation}
\K(G\mid H^*)\age0.
\label{eq:hyperpref_nonproj}
\end{equation}
Suppose, in addition, that the graph rule is at least as costly as reconstructing and running the hypergraph,
\begin{equation}
\K(D_G\mid G^*)\ageq\K(H\mid G^*)+\K(D_H\mid H^*).
\label{eq:hyperpref_rule}
\end{equation}
Expanding both joint complexities by the chain rule and applying Eq.~\eqref{eq:hyperpref_rule},
\[
\K(G,D_G)\aeq\K(G)+\K(D_G\mid G^*)\ageq\K(G)+\K(H\mid G^*)+\K(D_H\mid H^*).
\]
Symmetry of information~\eqref{eq:soi} gives $\K(G)+\K(H\mid G^*)\aeq\K(H)+\K(G\mid H^*)$, so
\[
\K(G,D_G)\ageq\K(H)+\K(G\mid H^*)+\K(D_H\mid H^*)\aeq\K(H,D_H)+\K(G\mid H^*).
\]
Hence
\begin{equation}
\Delta\RC\;\ageq\;\K(G\mid H^*)\;\age\;0,
\label{eq:hyperpref_delta}
\end{equation}
and the hypergraph model provides the shorter description, by a margin of at least the irreducible cost $\K(G\mid H^*)$ of recovering the skeleton from the hypergraph.

The interpretation is the mirror image of the graph-preferred case. There, the skeleton determined the groups, and the graph rule absorbed the higher-order information cheaply. Here, neither holds: the groups are genuine information not present in $G$, so a graph model must encode them explicitly inside its dynamical rule $D_G$. This forces a redundancy between $G$ and $D_G$ --- part of $D_G$ implicitly reconstructs the neighborhood structure already named in $G$ --- which inflates the joint description $(G,D_G)$ relative to the hypergraph, in which the groups are stated once as primitive objects. The surplus is exactly the non-projectable information measured by $\K(G\mid H^*)$.

\textit{Remark.} The condition $\K(G\mid H^*)\age0$ together with Eq.~\eqref{eq:hyperpref_rule} constitutes one sufficient mechanism by which a hypergraph representation can be shorter. A second, distinct mechanism arises when $G$ \emph{is} computable from $H$ (so that $\K(G\mid H^*)\aeq0$), but $H$ is \emph{not} computable from $G$: the graph rule must then re-encode the lost grouping information, inflating $\K(D_G\mid G^*)$ without any compensating reduction in $\K(G)$. This second mechanism is the operative one in Examples~\ref{ex:hyper_pref} and~\ref{ex:hyper_pref_xor}, where the clique expansion makes $G$ recoverable from $H$ but destroys the inverse recovery.

\subsubsection{Mixed regime}
\label{sec:mixed}

The remaining possibility combines a graph-favouring dynamical term with a hypergraph-favoring structural term: Assumption~1 tips toward the graph,
\[
\K(D_G\mid G^*)\ale\K(H\mid G^*)+\K(D_H\mid H^*),
\]
while Assumption~2 fails, so the skeleton is not recoverable from the hypergraph,
\[
\K(G\mid H^*)\age0.
\]
Carrying out the same chain-rule and symmetry-of-information steps as above, the relative cost decomposes as
\begin{equation}
\Delta\RC \;\aeq\; \underbrace{\bigl[\K(D_G\mid G^*)-\K(H\mid G^*)-\K(D_H\mid H^*)\bigr]}_{\ale\,0} \;+\; \underbrace{\K(G\mid H^*)}_{\age\,0}.
\label{eq:mixed_decomp}
\end{equation}
The two assumptions are exactly the signs of the two terms: a cheaper graph rule pushes $\Delta\RC$ down, while the non-recoverable skeleton pushes it up. Their sum is therefore not signed a priori, and the preferred representation is determined by which effect dominates for the particular system at hand --- the dynamical saving of the graph rule versus the structural information $\K(G\mid H^*)$ that the graph cannot inherit from the hypergraph.

This mixed regime is likely to be common in practice. A graph may support a simple aggregate or mean-field rule that avoids naming every group explicitly, while a hypergraph may encode observed groups that do not determine all pairwise or contextual information used by the graph. The graph is then favored if the dynamical compression in $D_G$ is larger than the structural residual $\K(G\mid H^*)$; the hypergraph is favored if the residual structural information dominates. Thus, the mixed case should be viewed as a quantitative trade-off rather than as a separate qualitative category.

\section{Approximating Representational Complexity}
\label{sec:approx}

The comparison developed above is exact but not operational: Kolmogorov complexity is incomputable, and the ideal representations $R_G^{\mathrm{ideal}}$, $R_H^{\mathrm{ideal}}$ that attain the minima in Eq.~\eqref{eq:restricted_R_ideal} cannot be exhibited in practice. To build intuition for how the regimes of Section~\ref{sec:comparison} can arise, we therefore turn from Kolmogorov complexity to concrete \emph{description-length} estimates.

For a structural language $\mathcal{L}$ and a dynamical class $\mathcal{C}$, let $\Len_\mathcal{L}(S)$ and $\Len_\mathcal{C}(D\mid S)$ denote the number of bits a fixed, explicit coding scheme spends on the structure $S$ and on the rule $D$ given $S$. Because any computable code is a valid description, every admissible pair gives an upper bound on the complexity of the ideal representation,
\begin{equation}
  \K(R^{\mathrm{ideal}})
  \;\aleq\;
  \min_{\substack{S\in\mathcal{L},\,D\in\mathcal{C}\\D(S)=F}}
  \bigl[\Len_\mathcal{L}(S)+\Len_\mathcal{C}(D\mid S)\bigr].
  \label{eq:L_upper_bound}
\end{equation}
In the examples below, we do not attempt this minimization. Instead, we fix one natural, explicitly constructed representation in each language --- not necessarily the ideal one --- and estimate its description length $L(\cdot)$ by a simple counting argument. Comparing the two estimates yields a proxy
\begin{equation}
  \Delta L := L(G,D_G)-L(H,D_H)
  \nonumber
\end{equation}
for the relative cost $\Delta\RC$, whose sign we read using the same three regimes as in Section~\ref{sec:comparison}: $\Delta L<0$ favors the graph, $\Delta L\approx0$ indicates parity, and $\Delta L>0$ favors the hypergraph, for the specific representations under comparison.

Two caveats must be kept in view, since they separate these examples from the theory above. First, $L$ is an upper-bound estimate for a \emph{chosen} encoding, not the Kolmogorov complexity of an ideal representation; $\Delta L$ is therefore a heuristic proxy for $\Delta\RC$, not an evaluation of it. Second, and consequently, each example compares two \emph{particular} representations and shows only which of those two is cheaper under the chosen codes. It does \emph{not} establish that graph or hypergraph modeling is preferable in general, nor that a graph-preferred or hypergraph-preferred regime exists at the level of Kolmogorov complexity. The examples illustrate circumstances under which a preference can appear in practice; they are not theorems about the formalism. 

Moreover, the sign of $\Delta L$ can depend on the chosen code. For example, a hypergraph language that explicitly lists all groups may be more expensive than a graph language, whereas a hypergraph language that admits the generating rule ``all $k$-subsets of each neighborhood'' as a short code may collapse the same example into the equivalence regime. Robust placement of a system in one regime therefore requires a modeling justification for the admissible codes or a compression procedure showing that the sign persists across reasonable encodings.
We present each example compactly, stating only the system, the two encodings, and the resulting $\Delta L$; the full counting derivations are collected in Appendix~\ref{app:examples}.

\subsection{Example I: generative groups (graph-preferred)}
\label{ex:graph_pref_subsets}

Consider a $d$-regular graph $G$ with a neighborhood-threshold rule that activates node $i$ when at least $\theta$ of its neighbors are active. A hypergraph representation of the same dynamics could, in principle, list all $k$-subsets of each neighborhood as operative groups ($1\le\theta\le k\le d$) and apply a group-threshold rule over them. As shown in Appendix~\ref{app:graph_pref_subsets}, the two rules are dynamically identical: the all-subsets group-threshold and the plain neighborhood-threshold produce the same map $F$. The graph therefore costs $L(G,D_G)=Nd\log N$, while the hypergraph, if its language requires an explicit edge list, must enumerate $\binom{d}{k}$ subsets of $k$ labels per node, giving $L(H,D_H)=N\binom{d}{k}\,k\log N$. Since the groups are fully generated from $G$ by a fixed rule, listing them explicitly is pure overhead, and $\Delta L = Nd\log N - N\binom{d}{k}\,k\log N < 0$ for $k\ge2$: the graph is the shorter description.

\subsection{Example II: irrelevant hyperedges (graph-preferred)}
\label{ex:graph_pref_independent}

Consider an SI-type spreading process $x_i(t+1)=\bigvee_{j\in A_i}x_j(t)$ on a graph $G$, alongside a hypergraph $H$ of $M$ size-$k$ edges drawn from an independent source --- say, external metadata such as group memberships or co-occurrence records --- that carries no information about the adjacency sets $A_i$ driving the dynamics. The graph stores the adjacency and names the OR rule at cost $L(G,D_G)=\sum_{i}|A_i|\log N$. The hypergraph must pay for its edge list ($Mk\log N$) and then, because $H$ tells it nothing about the adjacency, must re-encode the neighbor sets in full inside its rule, giving $L(H,D_H)=Mk\log N+\sum_{i}|A_i|\log N$. The two adjacency costs cancel in the comparison, leaving $\Delta L = -Mk\log N < 0$: the hyperedges contribute nothing to the dynamics and the graph is cheaper by exactly the cost of carrying them.

\subsection{Example III: pairwise structure (equivalence)}
\label{ex:equiv}

When the dynamics depends only on pairwise relations, a hypergraph representation that simply wraps each edge of $G$ in a two-element hyperedge is informationally identical to the graph. Both structures name the same $|E|$ pairs with the same two labels per pair, so $L(G,D_G)=L(H,D_H)=2|E|\log N$ and $\Delta L=0$. The two representations are mutually interconvertible by a fixed rule, so neither carries information the other lacks, and no informational argument can prefer one over the other. The same conclusion holds whenever the operative groups are reduced to singletons or pairs recoverable from $G$, as for standard pairwise contagion or linear threshold dynamics: using hypergraph language, there is a harmless relabelling, not a modeling commitment with an informational price tag.

\subsection{Example IV: hidden groups (hypergraph-preferred)}
\label{ex:hyper_pref}

Consider a $d$-regular graph $G$ in which each node $i$ has a single true interaction group $A_i$ --- one specific $k$-subset of its $d$ neighbors --- and activates by a threshold over that group alone. The hypergraph $H=\{A_i\}_{i=1}^N$ names these groups directly, so its cost is simply $L(H,D_H)=Nk\log N$ for the edge list plus the named rule. The graph, by contrast, records all $d$ neighbors of each node at cost $Nd\log N$, but the skeleton gives no indication of which $k$-subset is operative; there are $\binom{d}{k}$ candidates per node and the rule must single one out, adding $N\log\binom{d}{k}$ bits. The total graph cost is therefore $L(G,D_G)=Nd\log N+N\log\binom{d}{k}$, giving $\Delta L = N(d-k)\log N + N\log\binom{d}{k} > 0$: the hypergraph wins because it names the operative groups as primitives, while the graph pays both for redundant neighborhood information and for the group-selection the skeleton cannot express. The full counting argument is in Appendix~\ref{app:hyper_pref}.

\subsection{Example V: clique ambiguity (hypergraph-preferred)}
\label{ex:hyper_pref_xor}

Consider $N$ nodes interacting through $M$ random $k$-hyperedges ($k\ge3$) under the XOR--AND rule $x_i(t+1)=\bigvee_{\ell:\,i\in e_\ell}\bigwedge_{j\in e_\ell\setminus\{i\}}x_j(t)$, which activates a node when some group it belongs to is unanimously active. The natural graph encoding is the clique expansion, replacing each hyperedge with its $\binom{k}{2}$ pairwise edges at structural cost $Mk(k-1)\log N$. Because the clique expansion does not uniquely determine the original grouping~\cite{LaRockLambiotte2025}, the rule cannot recover the groups from the adjacency and must re-encode them explicitly, adding another $Mk\log N$ bits; the total is $L(G,D_G)=Mk^2\log N$. The hypergraph simply lists the $M$ groups and names the rule at $L(H,D_H)=Mk\log N$. The comparison gives $\Delta L = Mk(k-1)\log N > 0$: by projecting groups onto cliques, the graph destroys the grouping information and is then forced to re-encode it in the rule, paying for both the expanded adjacency and the reconstructed groups and making the clique edges redundant given what the rule already carries. This ambiguity is system-dependent: special hypergraph families whose groups are uniquely recoverable from their pairwise projections would not fall into this example.

\subsection{Summary}
Across these examples, the sign of $\Delta L$ tracks a single question: whether the higher-order structure is compressible into pairwise form under the chosen codes. The graph model is cheaper when the skeleton is a short generative code for the operative groups (Examples~\ref{ex:graph_pref_subsets},~\ref{ex:graph_pref_independent}); the hypergraph model is cheaper when the groups are residual information the skeleton cannot supply (Examples~\ref{ex:hyper_pref},~\ref{ex:hyper_pref_xor}); and the two coincide when each is recoverable from the other (Example~\ref{ex:equiv}). This mirrors the exact theory of Section~\ref{sec:comparison}: in the unrestricted setting, any missing structure can be absorbed into the rule, but once the structural and dynamical codes are fixed, the comparison becomes a genuine compression problem. We stress again that these are comparisons of particular encodings, not evaluations of $\Delta\RC$, and that a general criterion for which regime a given system and model class fall into remains open.

\section{Discussion}

We have introduced Representational Complexity (RC) as an information-theoretic measure of the excess description length incurred when a dynamical map is represented within a chosen structural and dynamical language. Specifically, RC compares the optimal description length achievable in a restricted model class with the Kolmogorov complexity $\K(F)$ of the induced dynamical map, which acts as a representation-independent floor. The framework leads to three main conclusions.

First, the equation-based representation establishes the theoretical floor. The quantity $\K(F)$ is the Kolmogorov complexity of the dynamical map itself, independent of any particular decomposition into structure and dynamics. No exact representation can describe the map in fewer bits. Explicit structure is therefore not justified because it lowers the intrinsic complexity of the map, but because it provides mechanistic organization, interpretability, empirical observability, or efficient description within a restricted modeling class~\cite{Torres2021,BickGross2023}. In this sense, structure is not free: it is a modeling commitment whose informational cost must be balanced against its scientific utility.

Second, unrestricted expressiveness is not a sufficient criterion for choosing a structural language. If the dynamical rule is allowed to be an arbitrary computable object, then missing structural information can always be absorbed into the dynamics. A graph representation can encode missing hyperedge information inside the rule, whereas a hypergraph representation can encode the same information structurally. These descriptions have equal total Kolmogorov complexity up to additive constants while differing only in how information is distributed between structure and dynamics. Thus, unrestricted equivalence does not imply modeling equivalence.

Third, restricting the admissible structures and dynamics makes the choice of structural language meaningful. Once these are constrained, different languages may incur different description lengths. We showed that graph and hypergraph representations are informationally equivalent only conditionally --- specifically, when the higher-order structure can be reconstructed from the graph through its dynamical rule, and the graph can in turn be recovered from the hypergraph by a fixed projection. When both conditions hold, neither representation enjoys a compressive advantage, and any preference between them must rest on non-informational grounds such as interpretability, mechanistic transparency, empirical observability, or modeling convenience. A concrete instance is provided by \cite{Lambiotte2019}, who show that higher-order Markov models capture temporal path dependencies that pairwise projections destroy; the preference for the higher-order representation there is grounded in dynamical fidelity rather than description length.

Within the restricted setting, we also showed how graph-preferred and hypergraph-preferred regimes arise once the equivalence conditions are relaxed. At the exact Kolmogorov level, Section~\ref{sec:comparison} gives conditional mechanisms, established through the chain rule and symmetry of information, under which the relative cost $\Delta\RC$ can favor one representation over the other. Section~\ref{sec:approx} then provides operational examples, based on explicit codes, that illustrate analogous effects in computable description-length estimates. The exact results identify the regimes in principle; the operational examples build intuition for how they can arise in practice. Kolmogorov complexity is incomputable, and the ideal representations attaining the minima in Eq.~\eqref{eq:restricted_R_ideal} cannot be constructed in practice, which is why the operational examples are indispensable as a complement to the formal theory.

The broader implication is that structural languages should be understood as modeling commitments rather than merely expressive devices. A graph-based model with sufficiently flexible node dynamics may reproduce higher-order behavior, but the higher-order structure is then encoded implicitly within the dynamical rule, whereas a hypergraph representation treats it as an explicit part of the model architecture. This distinction affects inference, parameterization, generalization, interpretability, and the relationship between model assumptions and empirical observations.

The distinction is particularly pronounced in empirical systems where interactions are observed directly as joint events---such as group conversations, co-authorships, biochemical complexes, ecological assemblages, or neural co-activations. In such cases, a hypergraph need not be viewed as an embellishment of an underlying graph, but rather as a direct representation of the observed interaction units. The relevant question is therefore not whether a sufficiently expressive graph-based model can reproduce the same dynamics, but whether doing so preserves the mechanistic and statistical structure of interest, rather than merely transferring that structure into hidden dynamical complexity. A concrete illustration is provided by \cite{Llabres2026}, who show that social impact models on hypergraphs admit an exact microscopic reduction to pairwise dynamics on a weighted projected network; in the nonlinear regime, the resulting weights are state-dependent, making the higher-order structure implicit in a dynamically complex rule---precisely the kind of informational transfer from structure to dynamics whose cost $(\Delta\RC)$ is designed to measure. These systems are plausible candidates for hypergraph-preferred descriptions, but establishing that claim empirically would require the computable surrogates discussed below.

Our framework does not imply that hypergraphs are universally superior to graphs, nor the converse. It provides a language for reasoning about the informational consequences of modeling choices under explicit restrictions, and it locates any representational advantage in the interaction between the system being modelled and the assumptions being imposed, rather than in a universal hierarchy of expressive power. A complementary perspective is provided by Lacasa~\cite{Lacasa2026}, who shows via Carleman linearization that nonlinear graph-based dynamics can be represented as linear dynamics on a richer combinatorial structure, recovering a structural benefit for higher-order representations that is orthogonal to expressiveness.
This also clarifies the relationship with Peixoto et al.~\cite{Peixoto2026}: their expressiveness claim is compatible with the unrestricted floor established here, since graphs can encode any finite dynamics once the rule is free to absorb arbitrary information. Our aim is therefore not to refute the unrestricted expressiveness claim, but to relocate the comparison to the restricted modeling setting in which the cost of absorbing information into the rule becomes part of the question. In our language, a graph model with sufficiently flexible learned node dynamics corresponds to a representation in which the graph rule absorbs information that a hypergraph model may instead place explicitly in the structure. If that information is cheaply encoded by the rule, the graph can be preferred; if encoding it inside the rule creates a high residual cost, the hypergraph can be preferred. Our contribution concerns the complementary question of what happens once restrictions motivated by scientific modeling are imposed, where informational cost becomes relevant, and the choice of language can matter. The purpose of $\RC$ is not to crown a universally preferred representation but to formalize how modeling constraints distribute description length.

The framework is also closely related to algorithmic statistics and the Kolmogorov structure-function program. There, one studies two-part descriptions that split the information in a data object into a model, or statistic, and a residual description of the data given the model. Our construction has the same algorithmic-information-theoretic spirit, but the split is specialized to networked dynamical systems: the two parts are an explicit interaction skeleton $S$ and a computable dynamical rule $D$, and the comparison is made between constrained structural languages such as graphs and hypergraphs. Thus, the novelty is not a new general theorem in algorithmic statistics, but the use of this decomposition to formulate a structural-language comparison for dynamical systems. The notions of interpretability, mechanistic transparency, and identifiability invoked above should be understood in this modeling sense; formalizing them as additional code-length or statistical-identifiability penalties would be a natural extension.

Finally, the exact Kolmogorov quantities studied here are not computable, so empirical application will require computable surrogates --- minimum-description-length (MDL) estimators or lossless compressors applied to explicit graph-, hypergraph-, and rule-based encodings. The connection with MDL is direct but not identical. Kolmogorov complexity provides the representation-independent theoretical baseline $\K(F)$ and the ideal restricted costs in Eq.~\eqref{eq:restricted_R_ideal}; MDL replaces these uncomputable objects with explicit code families and compares the resulting two-part descriptions. In that sense, the operational quantity $L(S,D)=L_{\mathcal{L}}(S)+L_{\mathcal{C}}(D\mid S)$ is an MDL-style surrogate for the restricted representation length.

This suggests a possible empirical program. One would specify admissible graph and hypergraph code families, including both structural codes and rule codes, fit or select the best representation within each family, and then compare their total two-part description lengths. Such a procedure would not compute $\RC$ exactly, but it would provide a reproducible test of whether the graph or hypergraph language yields a shorter description under the declared modeling assumptions. The counting estimates of Section~\ref{sec:approx} are a first, deliberately simple step in this direction; richer estimators would replace our fixed codes with genuine compression or MDL model-selection criteria. If graph-preferred or hypergraph-preferred regimes exist, such estimators should reveal corresponding differences in compressibility under appropriate modeling constraints. Developing them and applying them to empirical systems is an important direction for future work.

\section{Acknowledgements}
This work was supported by the European Union’s Horizon Europe Marie Sklodowska-Curie Actions
under the “BeyondTheEdge: Higher-Order Networks and Dynamics” project (Grant Agreement
No. 101120085). G.F.A. was supported by the \textit{Fundação de Amparo à Pesquisa do Estado de São Paulo} (FAPESP), Process Number 2024/16711-8 and 2025/04409-8. Y.M. was also partially supported by the Government of Arag\'on, Spain, and ERDF ``A way of making Europe'' through grant E36-23R (FENOL), and by Grant No. PID2023-149409NB-I00 from Ministerio de Ciencia, Innovaci\'on y Universidades, Agencia Espa\~nola de Investigaci\'on (MICIU/AEI/10.13039/501100011033) and ERDF ``A way of making Europe''.

\bibliographystyle{unsrt}
\bibliography{references}


\appendix

\section{Description-length derivations for the examples}
\label{app:examples}

This appendix gives the detailed counting arguments behind the description-length estimates of Section~\ref{sec:approx}. Throughout, we work with a single fixed coding scheme and report exact bit counts, so all relations are ordinary equalities and inequalities; no additive-constant relations and no Kolmogorov-complexity quantities appear. We emphasize once more that each $L(\cdot)$ is the length of one explicitly chosen, not necessarily optimal, encoding, so the resulting $\Delta L$ compares two particular representations and is only a heuristic proxy for $\Delta\RC$.

\paragraph*{Encoding conventions.}
We adopt the following uniform conventions, with all lower-order contributions (length prefixes, ceilings on $\log N$, parameter headers) suppressed.
\begin{itemize}
\item A single node label costs $\log N$ bits, where $\log\equiv\log_2$ and $N$ is the number of nodes.
\item An adjacency list specifying $|A_i|$ neighbours of node $i$ costs $|A_i|\log N$ bits; summed over nodes, a graph costs $\bigl(\sum_i|A_i|\bigr)\log N$ bits, i.e.\ $\log N$ per node--neighbour incidence.
\item A hyperedge on $k$ nodes costs $k\log N$ bits; a hypergraph of $M$ such edges costs $Mk\log N$ bits.
\item A \emph{named} dynamical rule --- one drawn from a fixed, system-independent catalogue (``OR over neighbors'', ``threshold $\geq\theta$'', and so on) --- costs a fixed number of bits independent of $N$, which we absorb into zero at the level of the leading-order counts: $L(D\mid S)=0$ for a named rule. This convention is not meant to make the rule literally costless; it only suppresses an $O(1)$ catalogue index that cannot affect the sign of the extensive examples. A rule that must additionally specify system-dependent data is charged for that data explicitly.
\end{itemize}
The relative cost reported in each example is $\Delta L = L(G,D_G)-L(H,D_H)$.

\subsection{A graph preferred case: example~\texorpdfstring{\ref{ex:graph_pref_subsets}}{IV 1}}
\label{app:graph_pref_subsets}

Let $G$ be $d$-regular with neighborhoods $U_i$ of size $d$, and fix $1\le\theta\le k\le d$. The hypergraph assigns to each node the family $H_i=\{A\subseteq U_i:|A|=k\}$ of all $k$-subsets of its neighborhood, and the dynamics is the group-threshold rule
\[
x_i(t+1)=\mathbf{1}\!\Bigl[\exists\,A\in H_i:\ \textstyle\sum_{j\in A}x_j(t)\ge\theta\Bigr].
\]
We first show this rule coincides with the plain neighborhood-threshold rule $x_i(t+1)=\mathbf{1}[\sum_{j\in U_i}x_j(t)\ge\theta]$. If $\sum_{j\in U_i}x_j(t)\ge\theta$, then $U_i$ contains $a\ge\theta$ active nodes; since $\theta\le k\le d$, we may choose a $k$-subset $A\subseteq U_i$ containing $\min(a,k)\ge\theta$ of them, so some $A\in H_i$ meets the threshold. Conversely, if some $A\in H_i$ satisfies $\sum_{j\in A}x_j(t)\ge\theta$, then $\sum_{j\in U_i}x_j(t)\ge\sum_{j\in A}x_j(t)\ge\theta$ because $A\subseteq U_i$. The two rules therefore produce the same map $F$.

\emph{Graph cost.} The structure is the $d$-regular adjacency, with $Nd$ incidences, hence $L(G)=Nd\log N$. The dynamics is the named neighborhood-threshold rule, so $L(D_G\mid G)=0$ and
\[
L(G,D_G)=Nd\log N.
\]

\emph{Hypergraph cost.} In a language that lists hyperedges explicitly, each node contributes $\binom{d}{k}$ subsets of $k$ labels, costing $\binom{d}{k}k\log N$ bits per node. The dynamics is the named group-threshold rule, so $L(D_H\mid H)=0$ and
\[
L(H,D_H)=N\binom{d}{k}\,k\log N.
\]

\emph{Comparison.}
\[
\Delta L = Nd\log N - N\binom{d}{k}\,k\log N
= N\log N\Bigl[d-\binom{d}{k}k\Bigr].
\]
For $k=1$ the bracket vanishes ($\binom{d}{1}\cdot1=d$); for $k\ge2$ one has $\binom{d}{k}k>d$ whenever $d>1$, so $\Delta L<0$. The all-subsets family is generated from $G$ by a fixed rule and adds nothing dynamically, so listing it explicitly is pure structural overhead and the graph is the shorter description. In practice this is the regime of higher-order groups that are real but dynamically inert: the groups exist and are even derivable from the pairwise structure, yet the dynamics responds only to their aggregate effect --- here the total number of active neighbors --- as in a quorum or threshold process where the outcome depends on how many contacts are active rather than on which particular group activates it.

\subsection{A graph preferred case: example~\texorpdfstring{\ref{ex:graph_pref_independent}}{IV 2}}
\label{app:graph_pref_independent}

Let $N$ binary nodes evolve by the SI-type rule $x_i(t+1)=\bigvee_{j\in A_i}x_j(t)$ on a graph $G$ with adjacency sets $A_i$, and let a hypergraph $H$ of $M$ size-$k$ edges be supplied independently of $G$ (for instance as external metadata), so that $H$ carries no information about the sets $A_i$.

\emph{Graph cost.} The structure is the adjacency, with $\sum_i|A_i|$ incidences, hence $L(G)=\sum_i|A_i|\log N$. The dynamics is the named OR rule, so
\[
L(G,D_G)=\sum_{i=1}^{N}|A_i|\log N.
\]

\emph{Hypergraph cost.} The edge list costs $L(H)=Mk\log N$. Because $H$ is independent of the adjacency, it provides no information about the sets $A_i$ that govern the dynamics; the rule must therefore encode those sets in full, at the same incidence cost as the graph structure,
\[
L(D_H\mid H)=\sum_{i=1}^{N}|A_i|\log N.
\]
Hence
\[
L(H,D_H)=Mk\log N+\sum_{i=1}^{N}|A_i|\log N.
\]

\emph{Comparison.}
\[
\Delta L = \sum_i|A_i|\log N - \Bigl(Mk\log N+\sum_i|A_i|\log N\Bigr) = -Mk\log N < 0.
\]
The hyperedges are dead weight: they neither compress the adjacency nor inform the rule, and the hypergraph representation simply carries the adjacency cost of the graph plus the gratuitous edge list. This is the common situation in which higher-order groupings come from a separate data source --- shared affiliations, co-authorship, or co-purchase records, say --- and are overlaid on a process that in fact spreads only along pairwise contacts; the annotations are genuine measurements but the wrong channel for this dynamics, so they cost bits to store while the rule must still recover the pairwise structure that actually drives the spreading. The same annotations would, of course, be indispensable were the contagion genuinely group-based, as in the hypergraph-preferred examples below.

\subsection{An equivalence case: example~\texorpdfstring{\ref{ex:equiv}}{IV3}}
\label{app:equiv}

Let the dynamics depend only on pairwise relations carried by the edges of $G$, with $|E|$ edges, and let $H$ be the size-two hypergraph whose hyperedges are exactly those pairs.

\emph{Graph cost.} The edge list costs $L(G)=2|E|\log N$ (two labels per edge), equivalently $\sum_i|A_i|\log N=2|E|\log N$ since each edge contributes two incidences. The dynamics is a named pairwise rule, so
\[
L(G,D_G)=2|E|\log N.
\]

\emph{Hypergraph cost.} The hypergraph lists $|E|$ size-two edges, again two labels each, so $L(H)=2|E|\log N$; the dynamics is the same named pairwise rule, giving
\[
L(H,D_H)=2|E|\log N.
\]

\emph{Comparison.}
\[
\Delta L = 2|E|\log N - 2|E|\log N = 0.
\]
The two structures are interconvertible by a fixed rule --- wrap each edge in a two-element hyperedge, or extract the pair from each hyperedge --- so the higher-order language stores nothing beyond the skeleton and neither representation has an advantage. The same conclusion holds whenever the operative groups reduce to singletons or pairs recoverable from $G$, which is what one finds for genuinely pairwise mechanisms such as linear threshold or standard contagion dynamics, where casting the model in hypergraph language is a harmless relabelling that changes neither the description length nor the mechanism.

\subsection{A hypergraph preferred case: example~\texorpdfstring{\ref{ex:hyper_pref}}{IV 4}}
\label{app:hyper_pref}

Let $G$ be $d$-regular, and suppose each node $i$ has a single true interaction group $A_i\subseteq U_i$ of size $k$, with threshold dynamics $x_i(t+1)=\mathbf{1}[\sum_{j\in A_i}x_j(t)\ge\theta]$ over that group. The hypergraph of operative groups is $H=\{A_i\}_{i=1}^N$.

\emph{Hypergraph cost.} The structure lists $N$ hyperedges of size $k$, so $L(H)=Nk\log N$. The dynamics is the named threshold rule, hence
\[
L(H,D_H)=Nk\log N.
\]

\emph{Graph cost.} The structure is the $d$-regular adjacency, $L(G)=Nd\log N$. The skeleton specifies each neighborhood $U_i$ but not which of its $k$-subsets is the active group $A_i$; there are $\binom{d}{k}$ candidates per node. A graph-based rule that cannot treat the group assignment as structural must therefore encode, for each node, which candidate is active, at $\log\binom{d}{k}$ bits:
\[
L(D_G\mid G)=N\log\binom{d}{k}.
\]
Hence
\[
L(G,D_G)=Nd\log N+N\log\binom{d}{k}.
\]

\emph{Comparison.}
\[
\Delta L = \Bigl(Nd\log N+N\log\tbinom{d}{k}\Bigr)-Nk\log N
= N(d-k)\log N + N\log\binom{d}{k} > 0,
\]
since $d\ge k$ and $\binom{d}{k}\ge1$. The group assignment is information the skeleton cannot supply, so the graph pays an extensive penalty --- the per-node group-selection cost $N\log\binom{d}{k}$ together with the difference $N(d-k)\log N$ between storing the full neighborhoods and storing only the operative groups --- and the hypergraph is the shorter description. This is the regime one meets when the operative groups are genuine latent structure not implied by the contact network, such as a specific reacting complex among a molecule's many potential partners or a particular coalition among a person's contacts, so that naming the groups directly is far cheaper than recording the full neighborhoods and then singling out the active subset within each.

\subsection{A hypergraph preferred case: example~\texorpdfstring{\ref{ex:hyper_pref_xor}}{IV 5}}
\label{app:hyper_pref_xor}

Let $N$ nodes interact through $M$ random $k$-hyperedges with $k\ge3$, under the XOR--AND rule
\[
x_i(t+1)=\bigvee_{\ell:\,i\in e_\ell}\ \bigwedge_{j\in e_\ell\setminus\{i\}}x_j(t).
\]

\emph{Hypergraph cost.} The structure lists $M$ size-$k$ edges, so $L(H)=Mk\log N$. The dynamics is the named rule (conjunction within each incident hyperedge, then disjunction across them), giving
\[
L(H,D_H)=Mk\log N.
\]

\emph{Graph cost.} The clique expansion replaces each hyperedge by its $\binom{k}{2}=\tfrac12 k(k-1)$ pairwise edges, for $M\binom{k}{2}$ edges in total, costing two labels each:
\[
L(G)=M\binom{k}{2}\cdot2\log N = Mk(k-1)\log N.
\]
The clique expansion does not determine the original grouping: distinct hyperedge families can induce the same pairwise adjacency~\cite{LaRockLambiotte2025}, so the sets $e_\ell$ are not recoverable from $G$ alone. A graph-based rule that cannot access latent hyperedge labels must therefore re-encode the grouping it needs, namely the $M$ size-$k$ edges, at
\[
L(D_G\mid G)=Mk\log N.
\]
Hence
\[
L(G,D_G)=Mk(k-1)\log N + Mk\log N = Mk^2\log N.
\]

\emph{Comparison.}
\[
\Delta L = Mk^2\log N - Mk\log N = Mk(k-1)\log N > 0 \quad (k\ge3).
\]
For $k=2$ the hyperedges are ordinary edges and the clique-expansion ambiguity disappears; the example is therefore genuinely higher-order only for $k\ge3$. The graph must store the pairwise neighborhoods and the operative groups separately. The pairwise edges are then redundant given the re-encoded groups, and this duplication makes $(G,D_G)$ the longer description. This is what happens whenever a genuinely many-body interaction is forced into pairwise form: the clique expansion of a group conversation, a joint biochemical reaction, or a co-authored work loses track of which units acted together, so the graph model must pay twice --- once for the projected edges and again to reconstruct the groups they came from.

\end{document}